\documentclass[aps,preprintnumbers,nofootinbib]{revtex4}
\usepackage{epsfig,graphics}
\usepackage{amsfonts,amssymb,amsmath}            % for math symbols.
\usepackage{pifont}                              % for cross symbol
\usepackage{ifthen}
\usepackage{amsthm}
\usepackage{subfigure}
\usepackage{multirow}                           % for table features
\def\identity{\leavevmode\hbox{\small1\kern-3.8pt\normalsize1}}

\newcommand{\comment}[1]{}
\newcommand{\ket}[1]{\left |  #1 \right\rangle}

\newcommand{\ketbra}[2]{|#1\rangle\!\langle#2|}
\theoremstyle{plain}

\newtheorem{theorem}{Theorem}
\newtheorem{lemma}{Lemma}
\newtheorem{corollary}{Corollary}

\theoremstyle{definition}
\newtheorem*{SC}{Security Condition}

\begin{document}

\title{The Impossibility Of Secure Two-Party Classical Computation}
\date[]{August 21, 2007}

\author{Roger \surname{Colbeck}}
\email[]{r.a.colbeck@damtp.cam.ac.uk}
\affiliation{Centre for Quantum Computation,
             DAMTP,
	     Centre for Mathematical Sciences,
             University of Cambridge,
             Wilberforce Road,
             Cambridge CB3 0WA, UK}
\affiliation{Homerton College,
             Hills Road,
             Cambridge CB2 8PH, UK}

\begin{abstract}
We present attacks that show that unconditionally secure two-party
classical computation is impossible for many classes of function.  Our
analysis applies to both quantum and relativistic protocols.  We
illustrate our results by showing the impossibility of oblivious
transfer.
\end{abstract}

\maketitle

\section{Introduction}
Consider two parties wishing to compute some joint function of their
data (two millionaires might wish to know who is richer, for example).
A secure computation of such a function is one for which the only
information the first party gets on the input of the second is that
implied by the outcome of the computation, and vice versa.

\comment{In a trusted-third-party (\sc{TTP}) computation model, both
parties give their input to the \sc{TTP} who evaluates the function
then returns to each party the relevant output.}

In this work, we focus on {\it unconditional} security, whereby we
seek to construct a protocol whereby the two mistrustful parties can
communicate in order to achieve the task.  Security will rely on a
belief in the laws of physics.  We allow each party to exploit the
properties of both quantum mechanics and relativity in order to
achieve security.  While the security benefits of the former are well
known, relatively little investigation has been made into the extra
security afforded by the latter.  One positive result in relativistic
cryptography is that it allows variable-bias coin tossing to be
realized \cite{CK1}.  In this paper, we show that even using both
relativistic and quantum protocols, there are a large class of
functions for which secure two-party computation is impossible.  A
discussion of relativistic cryptography can be found in Refs.\
\cite{Kent_relBC,CK1}.

We call a computation classical, in spite of it potentially relying on
quantum communication for its implementation, because its inputs and
outputs are classical data.

\begin{table}
\begin{tabular}{|l|l|c|c|l|}
%\hline
\hline
Zero-input&Deterministic &$ \checkmark $  &Trivial\\
        &Random       one-sided         &$ \checkmark $  &Trivial\\
        &Random       two-sided         &\checkmark   &Biased
$n$-faced die roll (see \cite{CK1} for discussion)\\
\hline
One-input &Deterministic &$\checkmark $  &Trivial\\
        &Random       one-sided         &\ding{55}$^*$&One-sided
variable-bias $n$-faced die roll (this paper)\\
        &Random       two-sided         &$\checkmark^*$ &Variable-bias
$n$-faced die roll cf.\ \cite{CK1}\\
%        &             &               &             &die roll\\
\hline
Two-input &Deterministic one-sided     &\ding{55}     &cf.\ \cite{Lo}\\
          &Deterministic two-sided     &\ding{55}$^*$ &This paper\\
          &Random       one-sided      &\ding{55}$^*$ &This paper\\
          &Random       two-sided      &\ding{55}$^*$ &This paper\\
%\hline
\hline
\end{tabular}
\caption{Functions computable with unconditional security in two-party
  computations using (potentially) both quantum and relativistic
  protocols.  \checkmark\ indicates that all functions of this type
  are possible, \ding{55} indicates that all functions of this type
  are impossible, $\checkmark^*$ indicates that some functions of this
  type are possible and all functions of this type are conjectured to
  be possible,
%(\ding{55}) indicates that some functions of this type are
%  impossible, 
  and \ding{55}$^*$ indicates that some functions of this type are
  impossible.}
\label{fns}
\end{table}

Two-party computations can be divided into several classes, depending
on the number of parties that receive the output (the {\it sidedness}
of the function) and whether the function is deterministic or random.
In the two-sided case, we will further specialize to {\it single
function} computations, where both parties receive identical outcomes.
What is presently known about such functions is summarized in Table
\ref{fns}.  For a longer introduction to secure two-party computation,
see Ref.\ \cite{CK1}.

In this paper, we will show the impossibility of various secure
two-party computations, by giving an explicit cheating attack.  A
summary of the argument is as follows.  In a classical computation,
each party is supposed to input one of a finite set of classical
values.  However, the impossibility of classical certification
\cite{Kent_certif} means that one party cannot detect when the other
inputs a superposition of such inputs.  By keeping all decisions at
the quantum level until the end of the protocol, we can model the
entire computation as unitary.  The insecurity of the computation then
follows because there exists a measurement on the output state
generated by the superposed input, which allows the cheating party to
better distinguish between the possible inputs of the other party than
if they had been honest.  In most cases, we have impossibility proofs
for the simplest non-trivial cases of each class of function.  We
discuss at the end of the paper the possible generalizations.

In this paper we consider {\it perfectly secure} protocols---i.e.\
those for which the probability of cheating is strictly zero.
Further, our protocols are {\it perfectly correct}; that is, the
probability of error is strictly zero in the case where both parties
are honest.  One would like to extend our results to cover the case
of protocols for which the probability of cheating and of error tend
to zero in the limit that some security parameter tends to infinity.

\section{Computational Model}
We use a black box model for secure computation.  A black box
represents an idealized version of a protocol.  It can be thought of
as an unbreakable box which has an input and output port for each
party.  It features an authentication system (e.g., an unalterable
label) so that each party can be sure of the function it computes.  An
appropriately constructed protocol will prescribe a sequence of
information exchanges mimicking the essential features of such a black
box.  If one of the parties deviates from the prescribed exchanges,
the protocol should abort.  The question of whether or not it is
possible to construct a protocol mimicking a given black box will not
be addressed\footnote{However, we do eliminate certain types of black
  box, e.g.\ ones that allow classical certification (see later).}.
Rather, we show that cheating is possible even if such black boxes do
exist.

Since in any real protocol all measurements can be delayed until the
end, we consider only black boxes which perform unitary operations.
The outcomes of such unitary operations are distributed amongst the
parties.  At the end of a classical computation they are measured to
generate the outcome.  For a general two-sided function, we consider
the unitary, $U_f$, such that
\begin{eqnarray}
\label{gc1}
U_f\ket{i}_A\ket{j}_B\ket{0}\ket{0}=\ket{i}_A\ket{j}_B\sum_k\alpha^k_{i,j}\ket{kk}_{AB},
\end{eqnarray}
where $\{\alpha^k_{i,j}\}$ depend on the function being computed, and
the index $k$ runs over all possible outputs.  $i$ and $j$ correspond
to Alice's and Bob's inputs respectively, and their output\footnote{Recall
that we have restricted to single function computations.} is $k$
which is read by measurement in an orthonormal basis.  Outcome $k$
occurs with probability $|\alpha^k_{i,j}|^2$.  If the function is
deterministic, then, for each $i$ and each $j$, $|\alpha_{i,j}^k|=1$
for one value of $k$, and is zero for all others.  More generally, the
unitary, $U'_f$ performing
\begin{eqnarray}
\label{gc2}
U'_f\ket{i}_A\ket{j}_B\ket{0}\ket{0}\ket{0}=\ket{i}_A\ket{j}_B\sum_k\alpha^k_{i,j}\ket{kk}_{AB}\ket{\psi_{i,j}^k}_{AB},
\end{eqnarray}
would be of use to compute such a function, where the final Hilbert
space corresponds to an ancillary system the black box uses for the
computation (and has arbitrary dimension).  In the protocol mimicking
such a box, this final state must be distributed between Alice and Bob
in some way, such that the part that goes to Bob, for instance,
contains no information on Alice's input.

If black boxes implementing such unitaries were to exist, then each
party has two ways of cheating.  The first is by inputting a
superposition of states into the protocol, rather than a member of the
computational basis as they should.  The second involves using a
different measurement on the output of the black box than that
dictated by the protocol.  It follows from the impossibility of
classical certification \cite{Kent_certif} that a real protocol cannot
prevent the first attack.  Under these attacks, insecurity of
functions under $U_f$ implies insecurity under $U'_f$, as we show
below.  Hence it is sufficient to consider only the former.

Consider the case where Alice makes a superposed input,
$\sum_ia_i\ket{i}$, rather than a single member of the computational
basis.  Then, at the end of the protocol, her reduced density matrix
takes either the form 
\begin{equation}
\label{sigj}
\sigma_j=\sum_{i,i',k}a_ia_{i'}^*\alpha_{i,j}^k(\alpha_{i',j}^k)^*\ketbra{i}{i'}\otimes\ketbra{k}{k}
\end{equation}
or
\begin{equation}
\sigma'_j=\sum_{i,i',k}a_ia_{i'}^*\alpha_{i,j}^k(\alpha_{i',j}^k)^*\ketbra{i}{i'}\otimes\ketbra{k}{k}\otimes{\text{tr}}_B\ketbra{\psi_{i,j}^k}{\psi_{i',j}^k},
\end{equation}
where the first case applies to $U_f$, and the second to $U'_f$.

Alice is then to make a measurement on her state in order to
distinguish between the different possible inputs Bob could have made,
as best she could.  We will show that there exists a trace-preserving
quantum operation that Alice can use to convert $\sigma'_j$ to
$\sigma_j$ for all $j$.  It follows that Alice's ability to distinguish
between $\{\sigma'_j\}_j$ is at least as good as her ability to
distinguish between $\{\sigma_j\}_j$.

In order that the protocol functions correctly when both Alice and Bob
are honest, we require
$\text{tr}_B\ketbra{\psi_{i,j}^k}{\psi_{i,j}^k}\equiv\rho^{i,k}$ to be
conditionally independent of $j$ given $k$ (otherwise Alice can gain
more information on Bob's input than that implied by $k$ by a suitable
measurement on her part of this state).  By expressing $\rho^{i,k}$ in
its diagonal basis,
$\rho^{i,k}=\sum_m\lambda^{i,k}_mU_A^{i,k}\ketbra{m}{m}_A(U_A^{i,k})^{\dagger}$,
we have
\begin{equation}
\ket{\psi_{i,j}^k}=\sum_m\sqrt{\lambda_m^{i,k}}U_A^{i,k}\ket{m}_A\otimes U_B^{i,j,k}\ket{m}_B,
\end{equation}
where $\{\ket{m}_A\}_m$ form an orthogonal basis set on Alice's system
and likewise $\{\ket{m}_B\}_m$ is an orthogonal basis for Bob's system.
Bob then holds
\begin{equation}
\text{tr}_A\ketbra{\psi_{i,j}^k}{\psi_{i,j}^k}=\sum_m\lambda_m^{i,k}U_B^{i,j,k}\ketbra{m}{m}_B(U_B^{i,j,k})^{\dagger}.
\end{equation}
This must be conditionally independent of $i$ given $k$, hence so must
$\lambda_m^{i,k}$ and $U_B^{i,j,k}$.  Thus
\begin{equation}
\ket{\psi_{i,j}^k}=\sum_m\sqrt{\lambda_m^k}(U_A^{i,k}\otimes
U_B^{j,k})\ket{m}_A\ket{m}_B.
\end{equation}
It hence follows that there is a unitary on Alice's system converting
$\ket{\psi_{i_1,j}^k}$ to $\ket{\psi_{i_2,j}^k}$ for all $i_1$, $i_2$,
and that, furthermore, this unitary is conditionally independent of
$j$ given $k$.  Likewise, there is a unitary on Bob's system
converting $\ket{\psi_{i,j_1}^k}$ to $\ket{\psi_{i,j_2}^k}$ for all
$j_1$, $j_2$, with this unitary being conditionally independent of $i$
given $k$.

Returning now to the case where Alice makes a superposed input.  The
final state of the entire system can be written
\begin{equation}
\sum_{i,k}a_i\alpha_{i,j}^k\ket{i}_A\ket{j}_B\ket{k}_A\ket{k}_B(U_A^{i,k}\ket{m}_A)(U_B^{j,k}\ket{m}_B).
\end{equation}
Alice can then apply the unitary
\begin{equation}
V=\sum_{i,k}\ketbra{i}{i}_A\otimes\openone_B\otimes\ketbra{k}{k}_A\otimes\openone_B\otimes
(U_A^{i,k})^{\dagger}\otimes\openone_B
\end{equation}
to her systems leaving the state as
\begin{equation}
\sum_{i,k}a_i\alpha_{i,j}^k\ket{i}_A\ket{j}_B\ket{k}_A\ket{k}_B\sum_m\sqrt{\lambda_m^k}\ket{m}_A(U_B^{j,k}\ket{m}_B).
\end{equation}
Alice is thus in possession of density matrix
\begin{equation}
\sum_{i,i',k}a_ia_{i'}^*\alpha_{i,j}^k(\alpha_{i',j}^k)^*\ketbra{i}{i'}\otimes\ketbra{k}{k}\otimes\rho_A^k,
\end{equation}
where $\rho_A^k=\sum_m\lambda_m^k\ketbra{m}{m}_A$.  On tracing out the
final system, we are left with $\sigma_j$ as defined by (\ref{sigj}).

We have hence shown that there is a trace-preserving quantum operation
Alice can perform which converts $\sigma'_j$ to $\sigma_j$ for all
$j$, and that this operation is conditionally independent of $j$ given
$k$.  Hence Alice's ability to distinguish between Bob's inputs after
computations of the type $U'_f$ is at least as good as her ability to
distinguish Bob's inputs after computations of the type $U_f$, and so,
under the type of attack we consider, insecurity of computations
specified by $U_f$ implies insecurity of those specified by $U'_f$.
We will therefore consider only type $U_f$ in our analysis.  An
analogous argument follows for the one-sided case, and likewise for
the deterministic cases (which are special cases of the
non-deterministic ones).

\bigskip
We now state the security condition that will be shown to be breakable
for a large class of computation.
\begin{SC}
  Consider the case where Bob is honest.  For a computation to be
  considered secure, there can be no input, together with a
  measurement on the corresponding output that gives Alice a better
  probability of guessing Bob's input than she would have gained by
  following the protocol honestly and making her most informative
  input.  This condition must hold for all forms of prior information
  Alice holds on Bob's input.
\end{SC}

Let us emphasize that the use of the black box model does not restrict
the scope of our proofs: these apply to all real protocols.  The model
is common to discussions of universal composability (see Section
\ref{disc}) and makes manifest that is sufficient for parties to behave
dishonestly only in the initial and final steps of any protocol in
order to break our security condition\footnote{In any case, if a
  protocol mimicks a black box correctly, then there is no scope for
  cheating during its implementation.}.

\section{Deterministic Functions}
We first focus on the deterministic case.  Lo showed that two-input
deterministic one-sided computations are impossible to compute
securely \cite{Lo}, hence only two-sided deterministic functions
remain\footnote{Lo did not consider relativistic cryptography, but his
results apply to this case as well \cite{CK1}.}.  There is a further
consideration when discussing deterministic functions that leads us to
restrict the class of functions further.

Suppose that the outcome of such a protocol leads to some real-world
consequence.  In the dating problem \cite{GottesmanLo}, for example,
one requires a secure computation of $k=i\times j$, where
$i,j\in\{0,1\}$.  If the computation returns $k=1$, then the protocol
dictates that Alice and Bob go on a date.  This additional real-world
consequence is impossible to enforce, although both Alice and Bob have
some incentive not to stand the other up, since this results in a loss
of the other's trust.  A cost function could be introduced to quantify
this, but since suitable cost assignments must be assessed case by
case, it is difficult to develop general results.  To eliminate such
an issue, we restrict to the case where the sole purpose of the
computation is to learn something about the input of the other party.
No subsequent action of either party based on this information will be
specified.

We say that a function is {\it potentially concealing} if there is no
input by Alice which will reveal Bob's input with certainty, and
vice versa.  If the aim of the computation is only to learn something
about the input of the other party, and if Bob's data is truly
private, he will not enter a secure computation with Alice if she can
learn his input with certainty.  We hence only consider potentially
concealing functions in what follows.  In addition, we will ignore
{\it degenerate} functions in which two different inputs are
indistinguishable in terms of the outcomes they afford.  If the sole
purpose of the computation is to learn something about the other
party's input, then, rather than compute a degenerate function, Alice
and Bob could instead compute the simpler function formed by combining
the degenerate inputs of the original.

An alternative way of thinking about such functions is that they
correspond to those in which there is no cost for ignoring the real
world consequence implied by the computation.  At the other extreme,
one could invoke the presence of an enforcer who would compel each
party to go ahead with the computation's specified action.  This would
have no effect on security for a given function (a cheating attack
that works without an enforcer also works with one) but introduces a
larger set of functions that one might wish to compute.  There exist
functions within this larger set for which the attack we present does
not work.

We specify functions by giving the matrix of outcomes.  For
convenience the outputs of the function are labelled with consecutive
integers starting with 0.  We consider functions that satisfy the
following conditions:
\begin{enumerate}
\item \label{cond1} (Potentially concealing requirement) Each row and
  each column must contain at least two elements that are the same.
\item \label{cond2} (Non degeneracy requirement) No two rows or
columns should be the same.
\end{enumerate}
For instance, if $i,j\in\{0,1,2\}$ (which we term a $3\times 3$
function), the function $f(i,j)=1-\delta_{ij}$ is
\begin{center}
\begin{tabular}{cc|ccc|}
\multirow{2}{*}{$f(i,j)$}&&\multicolumn{3}{|c|}{$i$}\\
&\;&\;0\;&\;1\;&\;2\;\\
\hline
\multirow{3}{*}{$j$}&0\;&$0$&$1$&$1$\\
&1\;&$1$&$0$&$1$\\
&2\;&$1$&$1$&$0$\\
\hline
\end{tabular}
.
\end{center}
\smallskip This function is potentially concealing, and non-degenerate.

We consider the case of $3\times 3$ functions.  We first give a
non-constructive proof that Alice can always cheat, and then an
explicit cheating strategy.

Let us assume that we have a black box that can implement the
protocol, i.e., that performs the following operation:
\begin{eqnarray}
U_f\ket{i}_A\ket{j}_B\ket{0}\ket{0}=\ket{i}_A\ket{j}_B\ket{f(i,j)}_A\ket{f(i,j)}_B.
\end{eqnarray}
The states $\{\ket{i}_A\}$ are mutually orthogonal, as are the members
of the sets $\{\ket{j}_B\}$, $\{\ket{f(i,j)}_A\}$ and
$\{\ket{f(i,j)}_B\}$.  This ensures that Alice and Bob always obtain
the correct output if both have been honest.  The existence of such a
black box would allow Alice to cheat in the following way.  She can
first input a superposition, $\sum_{i=0}^2 a_i\ket{i}_A$ in place
of $\ket{i}_A$.  Her output from the box is one of
$\rho_0,\rho_1,\rho_2$, the subscript corresponding to Bob's input,
$j$, where (using the shorthand ${\rm tr}_B(\ket{\Psi})\equiv{\rm
tr}_B(\ketbra{\Psi}{\Psi})$)
\begin{eqnarray}
\rho_j\equiv{\rm tr}_B\left(U_f\sum_{i=0}^2a_i\ket{i}_A\ket{j}_B\ket{0}_A\ket{0}_B\right).
\end{eqnarray}
Alice can then attempt to distinguish between these using any
measurement of her choice.

The main result of this section is the following theorem.
\begin{theorem}
  Consider the computation of a $3\times 3$ deterministic function
  satisfying conditions \ref{cond1} and \ref{cond2}.  For each
  function of this type, there exists a set of coefficients, $\{a_i\}$
  such that when Alice has a uniform prior distribution over Bob's
  inputs and she inputs $\sum_{i=0}^2a_i\ket{i}_A$ into the protocol,
  there exists a measurement that gives her a better probability of
  distinguishing the three possible ($j$ dependent) output states than
  that given by her best honest strategy.
\label{thm1}
\end{theorem}

\begin{table}
\begin{center}
\begin{tabular}{cc|ccc|}
\multirow{2}{*}{$f(i,j)$}&&\multicolumn{3}{|c|}{$i$}\\
&\;&\;0\;&\;1\;&\;2\;\\
\hline
\multirow{3}{*}{$j$}&0\;&\;$0$\;&\;$a$\;&\;$.$\;\\
&1\;&$0$&$b$&$.$\\
&2\;&$1$&$b$&$.$\\
\hline
\end{tabular}
\caption{This function can be taken as the most general $3\times 3$
  function satisfying conditions 1 and 2, where $a\neq b$, and $a=0$
  or $b=0$ or $b=1$.  The dots represent unspecified (and not
  necessarily identical) entries consistent with the conditions.}
\label{tab1}
\end{center}
\end{table}
\begin{proof}
We will rely on the following lemma.
\begin{lemma} 
All $3\times 3$ functions satisfying conditions 1 and 2 can be put in
the form of the function in Table \ref{tab1}.
\end{lemma}
\begin{proof}
The essential properties of any function are unchanged under
permutations of rows or columns (which correspond to relabelling of
inputs), and under relabelling of outputs.  In order that the function
is potentially concealing, there can be at most one column whose
elements are identical.  By relabelling the columns if necessary, we
can ensure that this corresponds to $i=2$.  Relabelling the outputs
and rows, if necessary, the column corresponding to $i=0$ has entries
$(f(0,0),f(0,1),f(0,2))=(0,0,1)$.  The column corresponding to $i=1$
then must have entries $(a,a,b)$ or $(a,b,b)$, with $a\neq b$.  In the
case $(a,a,b)$, the $i=2$ column must have the form $(c,d,d)$, for
$c\neq d$, in which case we can permute the $i=1$ and $i=2$ columns to
recover the form $a,b,b$ for the $i=1$ column.  Relabellings always
put such cases into forms with $a=0$ or $b=0$ or $b=1$.
\end{proof}

Suppose Alice inputs $\frac{1}{\sqrt{2}}\left(\ket{0}+\ket{1}\right)$
into a function of the form given in Table \ref{tab1}.  After tracing
out Bob's systems, Alice holds one of
\begin{eqnarray}
\rho_0&=&\frac{1}{2}\left(\ketbra{00}{00}+\delta_{a,0}\left(\ketbra{00}{10}+\ketbra{10}{00}\right)+\ketbra{1a}{1a}\right)\\
\rho_1&=&\frac{1}{2}\left(\ketbra{00}{00}+\delta_{b,0}\left(\ketbra{00}{10}+\ketbra{10}{00}\right)+\ketbra{1b}{1b}\right)\\
\rho_2&=&\frac{1}{2}\left(\ketbra{01}{01}+\delta_{b,1}\left(\ketbra{01}{11}+\ketbra{11}{01}\right)+\ketbra{1b}{1b}\right).
\end{eqnarray}
Measurement using the set $\{E_{i,k}=\ketbra{ik}{ik}\}$ in effect
reverts to an honest strategy.  The probability of correctly
guessing Bob's input using these operators is the same as that
for Alice's best honest strategy.  These operators can be combined to
form just three operators, $\{E_{j'}\}$ such that a result
corresponding to $E_{j'}$ means that Alice's best guess of Bob's input
is $j'$.  Then
\begin{eqnarray}
\label{E0}E_0&=&\alpha_1\ketbra{00}{00}+\delta_{a,0}\ketbra{10}{10}+\delta_{a,1}\ketbra{11}{11}+\delta_{a,2}\ketbra{12}{12}+\delta_{a,3}\ketbra{13}{13}\\
\label{E1}E_1&=&(1-\alpha_1)\ketbra{00}{00}+\alpha_2\delta_{b,0}\ketbra{10}{10}+\alpha_3\delta_{b,1}\ketbra{11}{11}+\alpha_4\delta_{b,2}\ketbra{12}{12}+\alpha_5\delta_{b,3}\ketbra{13}{13}\\
\label{E2}E_2&=&\openone-E_0-E_1,
\end{eqnarray}
where the $\{\alpha_l\}$ are arbitrary parameters, $0\leq\alpha_l\leq
1$, and do not affect the success probability.  We will show that such
a measurement is not optimal to distinguish between the corresponding
$\{\rho_j\}$.  This follows from an existing result in state
estimation theory, as stated in the following theorem
\cite{Holevo,Yuen&,Helstrom}.
\begin{theorem}
Consider using a set of $M$ measurement operators, $\{E_j\}$, to
discriminate between a set of $M$ states, $\{\rho_j\}$, which occur
with prior probabilities, $\{q_j\}$, where the outcome corresponding
to operator $E_j$ indicates that the best guess of the state is
$\rho_j$.  The set $\{E_j\}$ is optimal if and only if
\begin{eqnarray}
\label{con1}E_j\left(q_j\rho_j-q_l\rho_l\right)E_l&=&0\;\;\;\forall\;
j,l\\ \label{con2}\sum_j E_jq_j\rho_j-q_l\rho_l&\geq &
0\;\;\;\forall\; l.  
\end{eqnarray} 
\label{thm2}
\end{theorem}
In the case of uniform prior probabilities, Equations (\ref{con1}) and
(\ref{con2}) imply respectively
\begin{eqnarray}
\label{first}
\left(\alpha_1=0\quad{\rm or }\quad\alpha_2=0\quad{\rm or
}\quad b\neq 0\right)\quad{\rm and}\quad \left(\alpha_1=1\quad{\rm or }\quad
a\neq 0\right)\quad{\rm and}\\\nonumber \quad \left(\alpha_1=1\quad{\rm or
}\quad\alpha_2=1\quad{\rm or }\quad b\neq 0\right)\quad{\rm and}\quad
\left(\alpha_3=0\quad{\rm or }\quad b\neq 1\right),
\end{eqnarray}
and
\begin{eqnarray}
\nonumber\left(b=1\quad{\rm or }\quad\alpha_3\geq\frac{1}{4}\right)\quad{\rm and}\quad
\left(b=0\quad{\rm or }\quad\alpha_2(1-\alpha_1)\geq\frac{1}{4}\right)\quad{\rm
and}\quad 
\left(a=1\quad{\rm or }\quad\alpha_3=1\quad{\rm or }\quad b\neq
1\right)\quad{\rm and}\quad \\
\left(\alpha_1=0\quad{\rm or }\quad(b\neq 0\quad{\rm and}\quad a\neq
0)\right)\quad{\rm and}\quad
\left(\alpha_1=1\quad{\rm or }\quad
b\neq 0\quad{\rm or }\quad\alpha_2=0\right).
\end{eqnarray}

In addition, because the function is in the form given in Table
\ref{tab1}, we also have
\begin{eqnarray}
\label{last}\left(a=0\quad{\rm or }\quad b=0\quad{\rm or }\quad
b=1\right)\quad{\rm and}\quad a\neq b.
\end{eqnarray}
The system of equations (\ref{first}--\ref{last}) cannot be satisfied
for any values of $a,b,\{\alpha_k\}$.  Hence, the measurement
operators (\ref{E0}--\ref{E2}) are not optimal for discriminating
between Bob's inputs, so Alice always has a cheating strategy.
\end{proof}

Our proof of Theorem \ref{thm1} is non-constructive---we have shown
that cheating is possible, but not explicitly how it can be done.
Except in special cases (e.g., where the states $\{\rho_j\}$ are
symmetric), no procedure for finding the optimal POVM to distinguish
between states is known \cite{Chefles,JRF}.  Nevertheless, we have
found a construction based on the square root measurement
\cite{HJW,HW} that, while not being optimal, gives a higher
probability of successfully guessing Bob's input than any honest
strategy.

The strategy applies to the states, $\sigma_j$, formed when Alice
inputs $\frac{1}{\sqrt{3}}\left(\ket{0}+\ket{1}+\ket{2}\right)$.  The
set of operators are those corresponding to the square root
measurement, defined by
\begin{eqnarray}
E_{j'}=\left(\sum_j\sigma_j\right)^{-\frac{1}{2}}\sigma_{j'}\left(\sum_j\sigma_j\right)^{-\frac{1}{2}}.
\end{eqnarray}
One can verify, case by case, that this strategy affords Alice a
better guessing probability over Bob's input than any honest one for
all functions of the form of Table \ref{tab1}.  The Mathematica script
which we have used to check this is available on the world wide web
\cite{mathematica_script}.

\section{Non-Deterministic Functions}
\subsection{Two-sided case}
\begin{table}
\begin{center}
\begin{tabular}{cc|cc|}
\multirow{2}{*}{$p(0|i,j)$}&&\multicolumn{2}{|c|}{$i$}\\
&&\;0\;&\;1\;\\
\hline
\multirow{2}{*}{$j$}&0\;&\;$p_{00}$\;&\;$p_{10}$\;\\
&1\;&\;$p_{01}$\;&\;$p_{11}$\;\\
\hline
\end{tabular}
\caption{The entries in the table give the probabilities of output 0
  given inputs $i,j$.  For example, if both parties input 0, then the
  output of the function is 0 with probability $p_{00}$, and 1 with
  probability $1-p_{00}$.}
\label{function1}
\end{center}
\end{table}

Initially, we specialize to the case $i,j,k\in\{0,1\}$.  We specify
such functions via a matrix of probabilities as given in Table
\ref{function1}.  For the two-sided case, the relevant black box
implements the unitary, $U$, given by
\begin{eqnarray}
\label{U}U\ket{i}_A\ket{j}_B\ket{0}\ket{0}=\ket{i}_A\ket{j}_B\left(\sqrt{p_{ij}}\ket{00}_{AB}+\sqrt{1-p_{ij}}\ket{11}_{AB}\right).
\end{eqnarray}
Suppose that Alice has prior information about Bob's input such that,
from her perspective, he will input $0$ with probability $q_0$, and
$1$ with probability $q_1=1-q_0$.  The maximum probability of correctly
guessing Bob's input using an honest strategy is
\begin{equation}
\label{ph}
p_h=\max_i(\max_j(p_{ij}q_j)+\max_j((1-p_{ij})q_j)).
\end{equation}
Denote Alice's final state by $\rho_j$, where $j$ is Bob's input.  The
optimal strategy to distinguish $\rho_0$ and $\rho_1$ is successful
with probability \cite{Helstrom}
\begin{eqnarray}
\label{pc}
\frac{1}{2}\left(1+{\rm tr}\left|q_0\rho_0-q_1\rho_1\right|\right).
\end{eqnarray}
\begin{theorem}
Let Alice input $\frac{1}{\sqrt{2}}\left(\ket{0}+\ket{1}\right)$ and
Bob input $j$ into the computation given in (\ref{U}).  Let Alice
implement the optimal measurement to distinguish the corresponding
$\rho_0$ and $\rho_1$ and call the probability of a correct guess
using this measurement $p_c$. Then, for all
$\{p_{00},p_{01},p_{10},p_{11}\}$, there exists a value of $q_0$ such
that $p_c>p_h$, unless, 
\begin{enumerate}
\item $p_{00}=p_{10}$ and $p_{01}=p_{11}$, or 
\item $p_{00}=p_{01}$ and $p_{10}=p_{11}$. 
\end{enumerate}
\label{thm3}
\end{theorem}
The two exceptional cases correspond to functions for which only one
party can make a meaningful input.  We hence conclude that all
genuinely two-input functions of this type are impossible to compute
securely.
\begin{proof}
Take $q_0=1-\epsilon$. For sufficiently small $\epsilon>0$, (\ref{ph})
implies $p_h=q_0$.
We then seek $p_c$.  The eigenvalues of
$q_0\rho_0-q_1\rho_1$ are
%\begin{eqnarray}
%\frac{1}{4}\left((p_{00}+p_{10})q_0-(p_{01}+p_{11})q_1\pm\sqrt{((p_{00}+p_{10})q_0-(p_{01}+p_{11})q_1)^2+4(\sqrt{p_{01}p_{10}}-\sqrt{p_{00}p_{11}})^2q_0q_1}\right)\\
%\frac{1}{4}\left((\overline{p_{00}}+\overline{p_{10}})q_0-(\overline{p_{01}}+\overline{p_{11}})q_1\pm\sqrt{((\overline{p_{00}}+\overline{p_{10}})q_0-(\overline{p_{01}}+\overline{p_{11}})q_1)^2+4(\sqrt{\overline{p_{01}}\;\overline{p_{10}}}-\sqrt{\overline{p_{00}}\;\overline{p_{11}}})^2q_0q_1}\right),
%\end{eqnarray}
%where $\overline{p_{ij}}\equiv 1-p_{ij}$.
%For convenience, we rewrite these as
\begin{eqnarray}
\lambda_{\pm}&=&\frac{1}{4}\left(a(\{p_{i,j}\})\pm\sqrt{a^2(\{p_{i,j}\})+b(\{p_{i,j}\})}\right)\\
\mu_{\pm}&=&\frac{1}{4}\left(a(\{\overline{p_{i,j}}\})\pm\sqrt{a^2(\{\overline{p_{i,j}}\})+b(\{\overline{p_{i,j}}\})}\right),
\end{eqnarray}
where $a(\{p_{i,j}\})=(p_{00}+p_{10})q_0-(p_{01}+p_{11})q_1$,
$b(\{p_{i,j}\})=4(\sqrt{p_{01}p_{10}}-\sqrt{p_{00}p_{11}})^2q_0q_1$,
and $\overline{p_{ij}}\equiv 1-p_{ij}$.

For $\epsilon$ sufficiently small, we have $a\gg b>0$.  Using
$\sqrt{1+x}\leq 1+\frac{x}{2}$, we find,
$\lambda_+\geq\frac{1}{4}(2a(\{p_{i,j}\})+\frac{b(\{p_{i,j}\})}{2a(\{p_{i,j}\})}),$ 
$\lambda_-\leq -\frac{b(\{p_{i,j}\})}{8a(\{p_{i,j}\})},$ 
$\mu_+\geq\frac{1}{4}(2a(\{\overline{p_{i,j}}\})+\frac{b(\{\overline{p_{i,j}}\})}{2a(\{\overline{p_{i,j}}\})}),$
and $\mu_-\leq
-\frac{b(\{\overline{p_{i,j}}\})}{8a(\{\overline{p_{i,j}}\})}$, with
equality iff $b(\{p_{i,j}\})=0$ and $b(\{\overline{p_{i,j}}\})=0$.  We
hence have
$\frac{1}{2}\left(1+\text{tr}|q_0\rho_0-q_1\rho_1|\right)\geq q_0$
and so $p_c\geq p_h$, with equality iff $p_{00}=p_{10}$ and
$p_{01}=p_{11}$, or $p_{00}=p_{01}$ and $p_{10}=p_{11}$.
\end{proof}
The explicit form of the cheating measurement is given in
\cite{Helstrom}.

\subsection{One-sided case}
For one-sided computations of non-deterministic functions, Alice can
cheat without inputting a superposed state.  In this case, the black
box performs the unitary
\begin{eqnarray}
U\ket{i}_A\ket{j}_B\ket{0}=\ket{i}_A\ket{j}_B\left(\sqrt{p_{ij}}\ket{0}_A+\sqrt{1-p_{ij}}\ket{1}_A\right),
\end{eqnarray}
where the last qubit goes to Alice at the end of the protocol.  The
following theorem shows that such computations cannot be securely
implemented.
\begin{theorem}
Having made an honest input to the black box above, Alice's optimum
procedure to correctly guess Bob's input is not given by a
measurement in the $\{\ket{0},\ket{1}\}$ basis, except if 
$p_{ij}\in\{0,1\}$ for all $i$, $j$.
\label{thm4}
\end{theorem}
\begin{proof} 
From (\ref{con1}) of Theorem \ref{thm2}, if Alice inputs $i=1$, the
measurement operators $\{\ketbra{0}{0},\ketbra{1}{1}\}$ are optimal
only if
\begin{equation}
q_0\sqrt{p_{10}(1-p_{10})}=(1-q_0)\sqrt{p_{11}(1-p_{11})}.
\end{equation}
For this to hold for all $q_0$, we require that either $p_{11}=0$ or
$p_{11}=1$, and either $p_{10}=0$ or $p_{10}=1$.  Similarly, if Alice
inputs $i=0$, we require either $p_{01}=0$ or $p_{01}=1$, and either
$p_{00}=0$ or $p_{00}=1$, in order that the specified measurement
operators are optimal.
\end{proof}
These exceptions correspond to functions that are deterministic, so do
not properly fall into the class presently being discussed.  Many are
essentially single-input, hence trivial, and all such exceptions are
either degenerate or not potentially concealing.

Our theorem also has the following consequence.
\begin{corollary}
One-sided variable-bias coin tossing \cite{CK1} is impossible.
\end{corollary}
\begin{proof}
A one-sided variable bias coin toss is the special case where both
$p_{00}=p_{10}$ and $p_{01}=p_{11}$.  These cases are not exceptions
of Theorem \ref{thm4}, and hence are impossible.
\end{proof}

\subsection{Example:  The Impossibility Of Oblivious Transfer}
\label{OT_imposs}
Here we show explicitly how to attack a black box that performs
oblivious transfer when used honestly.  This is a second proof of its
impossibility in a stand-alone manner (the first being Rudolph's
\cite{Rudolph}). \footnote{Impossibility had previously been argued on
  the grounds that oblivious transfer implies bit commitment and hence
  is impossible because bit commitment is.  However, while this
  argument rules out the possibility of a composable oblivious
  transfer protocol, a stand-alone one is not excluded.}  The
probability table for this task is given in Table \ref{tabOT}.
\begin{table}
\begin{center}
\begin{tabular}{cc|cc|}
\multirow{2}{*}{$p(k|i)$}&&\multicolumn{2}{|c|}{$i$}\\
&&\;0\;&\;1\;\\
\hline
\multirow{3}{*}{$k$}
&0\;&\;$\frac{1}{2}$\;&\;0\;\\
&1\;&\;0\;&\;$\frac{1}{2}$\;\\
&?\;&\;$\frac{1}{2}$\;&\;$\frac{1}{2}$\;\\
\hline
\end{tabular}
\caption{Probability table for oblivious transfer.}
\label{tabOT}
\end{center}
\end{table}

In an honest implementation of oblivious transfer, Bob is able to
guess Alice's input with probability $\frac{3}{4}$.  However, the
final states after using the ideal black box are of the form
$\ket{\psi_b}=\frac{1}{\sqrt{2}}\left(\ket{b}+\ket{?}\right)$, where
$\ket{0}$, $\ket{1}$ and $\ket{?}$ are mutually orthogonal.  These are
optimally distinguished using the {\sc POVM} $(E_0,\openone-E_0)$,
where
\begin{equation}
E_0=\frac{1}{6}\left(\begin{array}{ccc}
2+\sqrt{3}&-1&1+\sqrt{3}\\
-1&2-\sqrt{3}&1-\sqrt{3}\\
1+\sqrt{3}&1-\sqrt{3}&2\\
\end{array}\right).
\end{equation}
This {\sc POVM} allows Bob to guess Alice's bit with probability
$\frac{1}{2}\left(1+\frac{\sqrt{3}}{2}\right)$, which is significantly
greater than $\frac{3}{4}$.

\section{Discussion} 
\label{disc}
We have introduced a black box model of computation, and have given a
necessary condition for security.  Even if such black boxes were to
exist as prescribed by the model, one party can always break the
security condition.  Specifically, by inputting a superposed state
rather than a classical one, and performing an appropriate measurement
on the outcome state, one party can always gain more information on
the input of the other than that gained using any honest strategy.  In
the case of deterministic functions, this attack has only been shown
to work if the function is non-degenerate and potentially concealing.
In the case where the sole purpose of the function is to learn
something about the other party's input, these are the only relevant
functions.

Our theorems deal only with the simplest cases of each class of
function.  However, the results can be extended to more general
functions as described below.

{\bf Larger input alphabets:} A deterministic function is impossible
to compute securely if it possesses a $3\times 3$ submatrix which is
potentially concealing and satisfies the degeneracy requirement.  This
follows because Alice's prior might be such that she can reduce Bob to
three possible values of $j$.  This argument does not rule out the
possibility of all larger functions, since some exist that are
potentially concealing without possessing a potentially concealing
$3\times 3$ subfunction.  Nevertheless, we conjecture that all
potentially concealing functions have a cheating attack which involves
inputting a superposition and then optimally measuring the outcome.

In the non-deterministic case, all functions with more possibilities
for $i$ and $j$ values possess $2\times 2$ submatrices that are ruled
out by the attacks presented, or reduce to functions that are
one-input.  Therefore, no two-party non-deterministic computations
with binary outputs can satisfy our security condition.

{\bf Larger output alphabets:} In the non-deterministic case, we
considered only binary outputs.  We conjecture that the attacks we
have presented work more generally on functions with a larger range of
possible outputs.

\bigskip 
We have not proven that the aforementioned attacks work for all
functions within the classes given in Table \ref{fns}, although we
conjecture this to be the case.  Furthermore, for any given
computation, one can use the methods presented in this work to verify
its vulnerability under such attacks.

We now briefly place our results within the context of universal
security definitions.  In classical cryptography, there are two common
models for universal security, one introduced by Canetti
\cite{Canetti} and the other by Backes, Pfitzmann and Waidner
\cite{PW,BPW}.  Recently, such frameworks have been extended for use
in quantum protocols \cite{Ben-OrMayers,Unruh,CGS}.  The idea is that
if a protocol is universally secure (or universally composable), then
it can be used as a subprotocol in any larger protocol.  The large
protocol can then be divided into subprotocols, each of which is
assumed to behave as a black box with a defined ideal
functionality\footnote{Or can alternatively be described via a trusted
  third party}.  The task of proving the larger protocol secure then
reduces to that of proving that the subprotocols correctly mimick
their ideals, together with an argument that the combination of the
ideals correctly performs the overall task.

Our results imply that there is no way to define an ideal suitable for
realizing secure classical computation in a quantum relativistic
framework.  Hence, without making additional assumptions, or invoking
the presence of a trusted third party, secure classical computation is
impossible using the usual notions of security.  The quantum
relativistic world, while offering more cryptographic power than both
classical and quantum non-relativistic worlds, still does not permit a
range of computational tasks.

One reasonable form of additional assumption is that the storage power
of an adversary is bounded\comment{\footnote{In fact, it is believed
that there is a physical principle that bounds the information
capacity of a region in terms of its surface area (see
\cite{Bekenstein} for a recent review), although this bound is so
large that it is difficult to imagine it ever being of use.}}.  The
so-called bounded storage model has been used in both classical and
quantum settings.  This model evades our no-go results because
limiting the quantum storage power of an adversary forces them to make
measurements (or discard potentially useful parts of the system).
This invalidates our unitary model of computation.  In the classical
bounded storage model, the adversary's memory size can be at most
quadratic in the memory size of the honest parties in order to form
secure protocols \cite{Cachin&,Ding&}.  However, if quantum protocols
are considered, and an adversary's quantum memory is limited, a much
wider separation is possible.  Protocols exist for which the honest
participants need no quantum memory, while the adversary needs to
store half of the qubits transmitted in the protocol in order to cheat
successfully \cite{Damgard&3}.

We further remark that the cheating strategy we present for the
non-deterministic case does not work for all assignments of Alice's
prior over Bob's inputs---there exist functions and values of the
prior for which it is impossible to cheat using the attack we have
presented.  This continues to be the case when we allow Alice to
choose amongst the most general superposed input states.  As a
concrete example, consider the set
$(p_{00},p_{01},p_{10},p_{11})=(\frac{47}{150},\frac{103}{150},\frac{8}{9},\frac{5}{9})$,
with $q_0=\frac{1}{2}$ in the two sided version. Hence, in practice,
there could be situations in which Bob would be happy to perform such
a computation, for example, if he was sure Alice had no prior
information over his inputs.

\acknowledgments 

I would like to acknowledge Adrian Kent and Robert K\"onig for useful
discussions.  This work was partly supported by the European Union
through the Integrated Project QAP (IST-3-015848), SCALA (CT-015714),
and SECOQC and by the QIP IRC (GR/S821176/01).

\end{document}